\documentclass{article}
\usepackage{ijcai25}

\usepackage{times}
\usepackage{soul}
\usepackage{url}
\usepackage[hidelinks]{hyperref}
\usepackage[utf8]{inputenc}
\usepackage[small]{caption}
\usepackage{graphicx}
\usepackage{amsmath}
\usepackage{amsthm}
\usepackage{amsbsy}
\usepackage{booktabs}
\usepackage{algorithm}
\usepackage{algorithmic}
\usepackage[switch]{lineno}
\usepackage{xfrac}
\usepackage{xspace}

\newcommand{\myOmit}[1]{}

\newcommand{\mymax}{\mbox{\rm max}}
\newcommand{\mymin}{\mbox{\rm min}}

\newcommand{\mymedian}{$\mbox{\sc Median}$\xspace}
\newcommand{\gmv}{$\mbox{\sc GenMedian}$\xspace}
\newcommand{\mylrm}{$\mbox{\sc Lrm}$\xspace}
\newcommand{\myleftmost}{$\mbox{\sc Leftmost}$\xspace}
\newcommand{\myrightmost}{$\mbox{\sc Rightmost}$\xspace}
\newcommand{\myem}{{\sc RandEnds}\xspace}
\newcommand{\mylrmtrunc}{\mbox{$\mbox{\sc Lrmt}$}\xspace}

\newcommand{\mymidornearest}{$\mbox{\sc MidOrNearest}$\xspace}
\newcommand{\mymidornearesttwo}{$\mbox{\sc MidOrNearestXY}$\xspace}

\newcommand{\myboundbox}{$\mbox{\sc MinBoundBox}$\xspace}
\newcommand{\myboundboxtrunc}{$\mbox{\sc MinBoundBox}^{*}$\xspace}
\newcommand{\mymediantwo}{$\mbox{\sc MedianXY}$\xspace}

\newcommand{\mylrmormaxmin}{\mbox{\sc LrmP}\xspace}
\newcommand{\myemormaxmintwo}{\mbox{\sc RandEnds2p}\xspace}
\newcommand{\mylrmtruncormaxmin}{\mbox{$\mbox{\sc LrmtP}$\mbox{}}\xspace}
\newcommand{\mymaxmin}{$\mbox{\sc MinMaxP}$\xspace}
\newcommand{\mymaxmintrunc}{$\mbox{\sc MinMaxP}_{\gamma}$\xspace}
\newcommand{\mymaxmintwo}{$\mbox{\sc MinMax2P}$\xspace}
\newcommand{\mymaxmintwotrunc}{$\mbox{\sc MinMax2P}_{\lambda}$\xspace}

\usepackage{latexsym}

\newtheorem{theorem}{Theorem}

\title{Mechanism Design for Facility Location using Predictions}

\author{
    Toby Walsh
    \affiliations
    AI Institute, UNSW Sydney
    \emails
    tw@cse.unsw.edu.au
}

\begin{document}

\maketitle

\begin{abstract}
 We study mechanisms for the facility location
  problem augmented with predictions of the optimal facility location.
  We demonstrate that an egalitarian viewpoint which considers
 {\em both} the maximum distance of any agent from the facility {\em and} the
  minimum utility of any agent provides important new insights
  compared to a viewpoint that just considers the maximum distance. 
  As in previous studies, we consider performance in terms of consistency (worst case 
  when predictions are accurate) and robustness (worst
  case 
  irrespective of the accuracy of predictions).
  \myOmit{
  For instance, in the single facility
  problem, the \mymaxmin\ mechanism
  was previously shown to be 1-consistent and 2-robust with respect to
  the maximum distance, and this is optimal. However, we show here that, while this mechanism is 1-consistent with respect to the optimal minimum utility, it has no bound on its robustness with respect to the minimum utility.}
By considering how mechanisms with predictions
can perform poorly, we design
  new mechanisms that are more robust.
  Indeed, by adjusting parameters, we demonstrate
  how to trade robustness
  for consistency.
  We go beyond the
  single facility problem by
  designing novel strategy proof mechanisms
  for locating two facilities
    with bounded consistency and robustness that
  use two predictions for where to locate the two facilities.
\end{abstract}

\section{Introduction}

In online algorithms, an elegant method to improve (worst-case)
performance is to
provide 
predictions about future inputs.
Such predictions might come from machine learning methods
applied to historical data. 
For example,
a cache scheduler has to decide which pages to evict
from the cache without knowing future requests for page access.
However, we can use machine learning to predict future cache requests,
improving performance
of the cache scheduler when these predictions are accurate.
Recently, researchers have proposed exploiting predictions 
in mechanism design,
arguing that they 
will transform the design and analysis of
mechanisms. 

Most relevant to this work, Agrawal {\it et al.}  \shortcite{abgou22} 
augmented various mechanisms for facility
location with predictions of the optimal location.
 Facility location
 is a classic problem 
 where we decide the location of
 a facility so as to minimize the distance of agents from the
 facility.
 It models a number of collective decision problems
 such as deciding the optimal room temperature for a class room,  the
 maintenance budget for an apartment complex,
 or the best location for a mobile phone tower. 
Our aim is to use predictions of the optimal facility location
 to provide better performance guarantees when predictions are
accurate (consistency) without sacrificing worst-case performance
when they are not (robustness).

We look here in more detail
at the mechanisms 
proposed in \cite{abgou22} that take account
of the predicted optimal location of the facility. 
We demonstrate the 
importance of the precise choice of objective. In particular, we show
that an {\bf egalitarian viewpoint} considering{ {\bf  {\em both} the
maximum distance} of any agent from the facility {\bf {\em and} the minimum
utility} 
provides new insights. 
Considering just the maximum distance focuses on problems where agents are close to facilities and distances are small. To achieve good approximation ratios, a mechanism must return high quality solutions on such problems. This ignores problems which are arguably more challenging where some agents are necessarily some distance from the nearest facility. By also considering
minimum utilities, we consider how well mechanisms
perform when some distances are necessarily large and utilities are small.
Our results also demonstrate the value of {\bf censoring extreme
predictions}. 
Insights from this study (such as the
value of censoring predictions)
could be useful in the
design of mechanisms with predictions in
other application domains such as fair division,
school choice or ad auctions.

\section{Related work}

There is a considerable literature
on augmenting algorithms with predictions to
improve worst-case performance
(see \cite{predictionsurvey} for a survey). 
Several recent surveys
also summarize the considerable literature on 
mechanism design
for facility location \cite{faclocsurvey,ijcai2021-596}. 
Starting with
Procaccia and Tennenholtz
\shortcite{approxmechdesign2}, 
much analysis of strategy proof mechanisms
for facility location
has focused on approximating 
the total and maximum distance that agents must
travel to the nearest facility
\myOmit{For example, 
the \mymedian\ mechanism 
returns the optimal total distance, 
and 2-approximates 
the maximum distance, and no other deterministic and strategy proof
mechanism can do better
\cite{ptacmtec2013}.
As a second example,
Aziz {\it et al.} \cite{acllwaaai2020}
consider strategy proof mechanisms that
minimize the total or maximum distance agents must travel subjet to capacity constraints
on how many agents each facility can serve.
}
(e.g. \cite{ft2010,egktps2011,proportional,ft2013,zhang2014,flplimit2,acllwaaai2020,coordmedian}). 
Indeed, one  recent survey
\cite{ijcai2021-596} describes
the design of strategy proof mechanisms which approximate
well the total or maximum distance that agents travel
as the ``classic setting'' for approximate mechanism design. 

However, some recent work on approximate mechanism
design for facility location has started to
consider other objectives
such as the utility of agents as this can uncover
fresh insight. For example,
Walsh \shortcite{wpricai21,wecai2024}
has looked at strategy proof mechanisms for facility location
optimizing both the maximum distance
and the minimum utility of agents.
As a second example,
Han {\it et al.} \shortcite{DBLP:conf/aaai/HanJA23}
look to optimize several objectives 
from the $l$-centrum family of metrics
(which includes total and maximum distance)
simultaneously. As a third example,
Aziz {\it et al.} \shortcite{propwine2022} identified
srategy proof mechanism that satisfy
proportional fairness which is a normative condition on
the utility of agents.
As a fourth example, Mei {\it et al.} \shortcite{flpdesire}
consider strategy proof mechanisms maximizing 
a normalized utility of agents called ``happiness''

For the online version of the facility location problem,
Jiang {\it et al.} \shortcite{onlineflp}
study online algorithms guided by predictions.
Such online algorithms must irrevocably assign each agent to an open
facility upon its arrival
or must decide to open a new facility (at cost) to which to assign it.
They provide a near-optimal online algorithm that
offers a smooth tradeoff between the prediction error and the
competitive ratio.
Here, by comparison, we do not have to make online decisions but
suppose the mechanism has access to location data for all agents
simultaneously. Almanza {\it et al.} \shortcite{onlineflpm} and
Fotakis {\it et al.} \shortcite{onlineflpf} 
also look at the online version of the facility location problem
with predictions, and propose algorithms with good competitive
ratios. 

For obnoxious facility location, where the goal of agents is to be as
far away from the facility 
as possible,
Istrate and Bonchis \shortcite{oflp} study mechanism design with predictions. They
present strategy proof mechanisms that
explore the tradeoff between robustness and consistency on
various metric spaces such as intervals, squares, circles, trees and hypercubes.

\section{Formal background}

In a facility location problem, we need to decide where to locate a facility
to serve a set of agents. 
We  consider $n$ agents located at $x_1$ to $x_n$.
We assume without loss of generality that $x_1 \leq \ldots \leq x_n$. 
A mechanism $f$ locates the facility at $y$. 
Formally, $f(x_1,\ldots,x_n) = y$. 
We let $d_i$ be the distance of agent $i$ to the facility:
$d_i =  | x_i - y|$.
As in a number of previous studies mentioned earlier, we assume that agents and facilities
are on the interval
$[0,1]$, and the utility of agent $i$ is $1-d_i$.
The interval could be $[a,b]$ supposing we normalise by
$b-a$. Other utility functions such as those
based on inverse square distance would be interesting
and are subject of our future work. We start here, however, 
with one of the simplest possible utility functions as it has been used
in prior work. 
%

Having agents and facilities lie on
an
interval 
is both practically and theoretically interesting. 
In practice, agents and facilities can be limited 
by physical constraints. 
For example,
when locating charging stations in 
a factory, robots and
charging stations might be limited to the factory.
As a second example,
when setting a thermostat, we
are limited by the boiler. 
As a third example, when locating a distribution centre, 
the centre might have to be on the fixed road network.
Thus many settings require locations
to be limited to an interval.
Restricting agents to
an interval
also limits the extent to which
agents can
misreport their location to gain advantage.
A fixed interval has been used in several 
recent studies (e.g. \cite{abs-2111-01566,flprevisit}). 

Our focus is on egalitarian mechanisms that look to minimize
the maximum distance any agent must travel or, equivalently,
to maximize the minimum
utility of any agent. When we consider approximation ratios of the
optimal solution, the
utility and distance viewpoints offer different insights.
Indeed, we will show
that the viewpoint of the
minimum utility of any agent provides an alternative
but useful
perspective that is complementary to that provided by maximum
distance. 

We consider mechanisms 
with good normative properties. One such property is
unanimity. 
A mechanism is {\em unanimous} iff the facility is located where all
agents agree. Formally $f$ is unanimous iff for any
$x$, we have $f(x,\ldots,x)=x$.
A simple fairness property is anonymity.
A mechanism is {\em anonymous} iff permuting the 
agents
does not change the outcome. Formally $f$ is anonymous iff for any
permutation $\sigma$, we have $f(x_{\sigma(1)}, \ldots,x_{\sigma(n)})=f(x_1,\ldots,x_n)$.
A mechanism is {\em Pareto efficient} iff we cannot
move the facility location to make one agent better
off without hurting other agents.
Formally $f$ is Pareto efficient iff
for any 
$x_j, \ldots, x_n$, there does not
exist a location $z$ and agent $i$
with 
$|x_i-z|  < |x_i-f(x_1,\ldots,x_n)|$
and $|x_j-z|  \leq |x_j-f(x_1,\ldots,x_n)|$ for all $j \in [1,n]$. 
Another important property is
resistance to manipulation. A mechanism is {\em strategy proof} 
iff no agent can mis-report their location and reduce their
distance to
the nearest facility.
Formally $f$ is strategy proof iff for any 
$x_1,\ldots,x_n$, and any agent $i$, it is not the
case that there exists $x_i'$ with
$|x_i - f (x_1,\ldots,x_i',\ldots,x_n)| < |x_i-f(x_1,\ldots,x_i,\ldots,x_n)|$. 
We will consider how well strategy proof mechanisms
approximate an objective like the optimal maximum distance
or minimum utility. A mechanism has an {\em approximation ratio} $\rho$ 
for a maximization (minimization) objective
iff the answer returned
is at least $\sfrac{1}{\rho}$ (at most $\rho$) times the optimal. 

We consider a number of strategy proof mechanisms. %
Many are based on the 
function $median(z_1, \ldots, z_p)$ 
which returns $z_i$ where $|\{ j | z_j < z_i\}| < \lceil \sfrac{p}{2} \rceil$
and $|\{ j | z_j > z_i\}| \leq \lfloor \sfrac{p}{2} \rfloor$. 
For example, the \gmv\ mechanism locates a facility at
$median(x_1, \ldots, x_n, z_1, \ldots, z_{n-1})$
where the $n-1$ parameters
$z_1$ to $z_{n-1}$ are ``phantom'' agents at fixed locations.
Moulin  \shortcite{moulin1980} proved that
a mechanism 
is anonymous, Pareto efficient and strategy-proof
iff it is 
\gmv. 
The \myleftmost\ 
mechanism is an instance of \gmv\ with
$z_i = 0$ for $i \in [1,n)$, locating the facility
at the leftmost
agent. 
The \myrightmost\
mechanism is an instance of \gmv\ with
$z_i=1$ for $i \in [1,n)$, locating the facility
at the rightmost agent. 
The \mymedian\
mechanism is an instance of \gmv\ with 
$z_i = 0$ for $i \leq \lfloor
\sfrac{n}{2} \rfloor$ and $1$ otherwise,
locating the facility at the median agent. 
The \mymidornearest\ mechanism is an instance of \gmv\ with 
$z_i = \sfrac{1}{2}$ for $i \in [1,n)$. 
It locates the facility either 
at $\sfrac{1}{2}$
if $x_1 \leq \sfrac{1}{2} \leq
x_n$, otherwise at the agent
nearest to $\sfrac{1}{2}$. 

In \cite{abgou22}, mechanisms for facility location 
are augmented with a prediction $\pi$ of the optimal facility location.
For example, if $x_1$ and $x_n$ are the minimum and maximum reported locations
of the agents, the prediction augmented mechanism $\mymaxmin(x_1,x_n,\pi)$ returns 
the predicted solution
$\pi$ as facility location when $x_1 \leq \pi \leq x_n$, otherwise
it returns $x_1$ when $\pi< x_1$,
and $x_n$ when $\pi > x_n$.
In general, our goal is for predictions
to improve the performance 
when accurate and not to hinder
performance when inaccurate.
A mechanism with prediction is {\em $\alpha$-consistent}
with respect to
maximum distance/minimum utility
iff, when the prediction is correct,
the mechanism has an approximation ratio of $\alpha$ or better
with respect to the objective of
maximum distance/minimum utility. 
A mechanism with prediction is {\em $\beta$-robust} with respect to
 maximum distance/minimum utility
iff, irrespective of the quality of
the prediction, 
the mechanism has an approximation ratio of $\beta$ or better
with respect to the objective of
maximum distance/minimum utility. 
The \mymaxmin\ mechanism
is strategy proof and, 
with respect to maximum distance, is 
1-consistent 
and 2-robust 
 (i.e. returns the optimal maximum distance when the prediction is
 correct, and 2-approximates it otherwise)
\cite{abgou22}.

\section{Single facility}

We begin our study with the simplest setting
where we locate a single facility on the interval $[0,1]$.
You might think, based on the analysis 
of Agrawal {\it et al.}  \shortcite{abgou22} considering
approximating the maximum distance, that the \mymaxmin\ mechanism 
was optimal and all that could be usefully said
about strategy proof mechanisms that exploit predictions. 
This mechanism 
has optimal consistency and robustness with respect to the
maximum distance. 
Clearly no mechanism can do better than
1-consistency, while no deterministic and strategy proof mechanism can be
better than 2-robustness (Procaccia and Tennenholts
\shortcite{ptacmtec2013}
demonstrate this for the real line but the result
easily extends to any fixed interval).
%
%

An analysis of minimum utilities shows that there is more to uncover about egalitarian
mechanisms exploiting predictions.
Consider the 
subtly different egalitarian
objective of the minimum utility, and the
approximation ratios that can be achieved of
this objective. 
The \mymaxmin\ mechanism
is {\bf far from optimal from this perspective}. 
In fact, there is no bound on
how badly it approximates the 
minimum utility.

\begin{theorem}
  The \mymaxmin\ mechanism is 1-consistent with respect to the optimal
  minimum utility, but has no bound on its robustness. 
\end{theorem}
\begin{proof}
If the prediction is accurate then, as the mechanism is strategy
proof, $x_1 \leq \pi = \frac{x_1+x_n}{2} \leq x_n$. The facility
is therefore located at this accurate prediction, and the mechanism is 1-consistent.
For robustness, suppose $\pi=x_1=0$ and $x_n=1$.
Then the facility is located at 0, giving a minimum utility of
zero. However, the optimal minimum utility is $\sfrac{1}{2}$ with
the facility at $\sfrac{1}{2}$. Hence
robustness is unbounded. 
\end{proof}

Considering the approximation ratio of the optimal maximum
distance focuses attention
on problem instances where distances are small and all agents are
necessarily
close to the facility location. It ignores those more challenging
problem
instances where distances are large and some agents are
necessarily far from the facility location.
Unfortunately the \mymaxmin\ mechanism
may approximate poorly certain instances in which agents must
travel large distances. 

We compare this lack of robustness with the simple
strategy proof \mymidornearest\ mechanism.
This has consistency and robustness that is bounded with respect
to both maximum distance and minimum utility. 

\begin{theorem}
  The \mymidornearest\ mechanism is 
  $\sfrac{3}{2}$-consistent
  and $\sfrac{3}{2}$-robust with respect to the optimal
  minimum utility.
  It is
   ${2}$-consistent
   and ${2}$-robust with respect to the optimal
   maximum distance. 
\end{theorem}
\begin{proof}
  \mymidornearest\
  ignores the prediction so consistency is the same
as robustness. 
Theorem 1 in \cite{wecai2024} demonstrates
that the
mechanism
$\sfrac{3}{2}$-approximates the 
minimum utility
and $2$-approximates the 
maximum distance. 
\end{proof}

Note that no deterministic
and strategy proof mechanism can do
better than
$\sfrac{3}{2}$-robustness with respect to the 
minimum utility \cite{wecai2024}, 
or 2-robustness with
respect to the maximum distance \cite{ptacmtec2013}.
The \mymidornearest\ mechanism
is actually an instance of the \mymaxmin\ mechanism when the predicted
optimal location is $\sfrac{1}{2}$. It is extreme predictions away
from $\sfrac{1}{2}$ that lead to the lack of robustness of the
\mymaxmin\ mechanism.
\myOmit{Note that the even simpler \mymedian\ mechanism achieves 2-consistency and
2-robustness
with respect to the maximum distance but has an unbounded
consistency and robustness with respect to minimum utility. 
}

Of course,
the \mymidornearest\ mechanism is not exploiting any information about
the predicted optimal facility location. 
Mechanisms
which
are responsive to the predicted facility location 
can do better. 
However, to get good (bounded) robustness with respect to the
minimum utility, we must avoid extreme predictions
near the interval end points.
We propose next a {\bf new mechanism} guided
by {\bf non-extreme predictions} that has bounded robustness.

The \mymaxmintrunc\ mechanism is a truncated
version of the \mymaxmin\ mechanism with a parameter
$\gamma \in [0,\sfrac{1}{2}]$. It maps the prediction $\pi$
onto $\mymax(\gamma,\mymin(\pi,1-\gamma))$, and then
applies the \mymaxmin\ mechanism to this truncated prediction.
This mapping limits predictions to 
the interval $[\gamma,1-\gamma]$.
The \mymaxmintrunc\ mechanism is 
the \mymaxmin\ mechanism when $\gamma=0$,
and the \mymidornearest\ mechanism when $\gamma = \sfrac{1}{2}$.
For $0 < \gamma < \sfrac{1}{2}$, it is a synthesis
of the two mechanisms.
smoothly interpolating
between \mymaxmin\  (which has
optimal consistency) and
 \mymidornearest\  (which, as
argued shortly, has optimal robustness). 

\begin{theorem}
  For  $\gamma \in [0,\sfrac{1}{2}]$,
  the \mymaxmintrunc\ mechanism is strategy proof,
$\frac{(2-\gamma)}{(2-2\gamma)}$-consistent
  and $\frac{(1+\gamma)}{2\gamma}$-robust
  with respect to the optimal
  minimum utility.
  It is ${1}$-consistent with respect to the optimal maximum
   distance when $\gamma = 0$,
   but $2$-consistent when $\gamma>0$.
   It is always $2$-robust with respect to the optimal maximum
   distance. 
   \end{theorem}
\begin{proof}
Strategy proofness is immediate from that of the untruncated
mechanism.
With respect to the optimal minimum utility,
suppose the prediction $\pi$ is correct.
There are five cases. In the first case, $\pi \leq
\sfrac{\gamma}{2}$. Let the minimum utility be $1-b$
with $b \leq \pi$.
The 
mechanism
locates the facility at $\pi+b$ giving a
minimum utility of $1-2b$.
The approximation ratio is thus $\frac{(1-b)}{(1-2b)}$.
This has a maximum of $\frac{(2-\gamma)}{(2-2\gamma)}$
when $\pi = b = \sfrac{\gamma}{2}$. 
In the second case, $\sfrac{\gamma}{2} \leq \pi \leq \gamma$ and
the minimum utility is $1-b$
with $b \leq \gamma - \pi$. 
The 
mechanism
locates the facility at $\pi+b$ giving a
minimum utility of $1-2b$.
The approximation ratio is thus $\frac{(1-b)}{(1-2b)}$.
This again has a maximum of $\frac{(2-\gamma)}{(2-2\gamma)}$
when $\pi = b = \sfrac{\gamma}{2}$. 
In the third case, $\sfrac{\gamma}{2} \leq \pi \leq \gamma$ and
the minimum utility is $1-b$
with $\pi \geq b \geq \gamma - \pi$. 
The 
mechanism
locates the facility at $\gamma$ giving a
minimum utility of $1-(\gamma-(\pi-b))$.
The approximation ratio is thus $\frac{(1-b)}{(1-b-\gamma+\pi)}  $.
This again has a maximum of $\frac{(2-\gamma)}{(2-2\gamma)}$
when $\pi = b = \sfrac{\gamma}{2}$. 
In the fourth case, $\gamma \leq \pi \leq 1-\gamma$.
The 
mechanism
locates the facility at $\pi$ giving the optimal minimum
utility and an approximation ratio of 1.
In the fifth case, $\pi \geq 1-\gamma$. This is symmetric
to the first three cases.
Over the five cases, the largest approximation ratio is
$\frac{(2-\gamma)}{(2-2\gamma)}$. 

Now 
suppose the prediction
is incorrect.
There are four cases.
In the first case, $x_n$ is less than or equal to 
the truncated prediction $\pi'$.
The optimal minimum utility is $1-\frac{(x_n-x_1)}{2}$.
However, the 
mechanism locates
the facility at $x_n$ giving a minimum utility of $1-(x_n-x_1)$.
The approximation ratio is therefore
$\frac{(2-x_n+x_1)}{2(1-x_n+x_1)}$.
This is maximized for $x_1=0$ and $x_n=1-\gamma$
when the ratio is $\frac{(1+\gamma)}{2\gamma}$. 
In the second case, 
$\pi'$ is between
$x_1$ and $x_n$ or equal to $x_1$,
and nearer to $x_1$ than $x_n$.
The optimal minimum utility is
again $1-\frac{(x_n-x_1)}{2}$.
However, the 
mechanism locates
the facility at 
$\pi'$. 
The minimum utility is $1-(x_n-\pi')$.
The approximation ratio is therefore
$\frac{(2-x_n+x_1)}{2(1-x_n+\pi')}$.
This is maximized for $x_1=0$, $x_n=1-\gamma$
and $\pi'=\gamma$ 
when the ratio is $\frac{(1+\gamma)}{2\gamma}$. 
In the third case, 
$\pi'$ is between
$x_1$ and $x_n$ or equal to $x_n$, and not
nearer to $x_1$ then $x_n$.
This is symmetric to the second case. 
In the fourth
case, 
$\pi'$ is greater than $x_n$.
This is symmetric to the first case.
Over the four cases, the largest approximation ratio is
$\frac{(1+\gamma)}{2\gamma}$. 

With respect to the optimal maximum distance,
suppose the prediction $\pi$ is correct. For $\gamma=0$,
the \mymaxmintrunc\ mechanism is equivalent to
\mymaxmin\ mechanism which is $1$-consistent and $2$-robust.
For $\gamma>0$ there are five cases. In the first case, $\pi \leq
\sfrac{\gamma}{2}$. Let the maximum distance be $b$
with $b \leq \pi$.
The 
mechanism
locates the facility at $\pi+b$ giving a
maximum distance of $2b$.
The approximation ratio is thus $2$.
In the second case, $\sfrac{\gamma}{2} \leq \pi \leq \gamma$ and
the maximum distance is $b$
with $b \leq \gamma - \pi$. 
The 
mechanism
locates the facility at $\pi+b$ giving a
maximum distance of $2b$.
The approximation ratio is thus again $2$.
In the third case, $\sfrac{\gamma}{2} \leq \pi \leq \gamma$ and
the maximum distance is $b$
with $\pi \geq b \geq \gamma - \pi$. 
The 
mechanism
locates the facility at $\gamma$ giving a
maximum distance of $(\gamma-(\pi-b))$.
The approximation ratio is thus $\frac{(b+\gamma-\pi)}{b}$.
This has a maximum of $2$ 
when $\pi = b = \sfrac{\gamma}{2}$. 
In the fourth case, $\gamma \leq \pi \leq 1-\gamma$.
The 
mechanism
locates the facility at $\pi$ giving the optimal maximum
distance and an approximation ratio of 1.
In the fifth case, $\pi \geq 1-\gamma$. This is symmetric
to the first three cases.
Over the five cases, the largest approximation ratio is
$2$. 

Now suppose again that the prediction
is incorrect.
There are four cases.
In the first case, $x_n$ is less than or equal to the
truncated prediction $\pi'$. 
The optimal maximum distance is $\frac{(x_n-x_1)}{2}$.
However, the 
mechanism locates
the facility at $x_n$ giving a maximum distance of $(x_n-x_1)$.
The approximation ratio is therefore
$2$. 
In the second case, 
$\pi'$ is between
$x_1$ and $x_n$ or equal to $x_1$,
and nearer to $x_1$ than $x_n$.
The optimal maximum distance is
again $\frac{(x_n-x_1)}{2}$.
However, the 
mechanism locates
the facility at the truncated prediction $\pi'$. 
The maximum distance is $x_n-\pi'$.
The approximation ratio is therefore
$\frac{2(x_n-\pi')}{(x_n-x_1)}$.
This is maximized for $x_1=\gamma$, $x_n=1$ 
and $\pi'=\gamma$ when the ratio is $2$. 
In the third case, 
$\pi'$ is between
$x_1$ and $x_n$ or equal to $x_n$, and not
nearer to $x_1$ than $x_n$.
This is symmetric to the second case. 
In the fourth
case, 
$\pi'$ is greater than $x_n$.
This is symmetric to the first case.
Over the four cases, the largest approximation ratio is
$2$. 
\end{proof}

Note that when the prediction is in $[\gamma,1-\gamma]$, the
\mymaxmintrunc\ mechanism does even better. In this setting, the mechanism
is 1-consistent with respect to minimum utility or maximum distance.
It is only with extreme predictions (less than $\gamma$ or greater
than $1-\gamma$) where consistency drops.
\myOmit{
Note also that while the consistency of the 
\mymaxmintrunc\ mechanism with respect to the minimum utility
drops smoothly from 1 to $\sfrac{3}{2}$
  as $\gamma$ goes from 0 to $\sfrac{1}{2}$, the consistency
  is discontinuous with respect to maximum distance,
  jumping from 1-consistent at $\gamma=0$ to 2-consistent
  for any $\gamma > 0$. The additive error in distance increases
  continuously with $\gamma$ but the multiplicative error
  is unbounded since the optimal maximum distance may be zero. }
Note also that by adjusting 
  $\gamma$, we can {\bf trade 
    consistency for robustness} (see Figure 1 for a visualization
  of this).
   At $\gamma=0$, the 
\mymaxmintrunc\  mechanism
   is 1-consistent with respect to minimum utility but has unbounded robustness.
   Increasing $\gamma$ decreases robustness
   but increases consistency.
   At $\gamma=\sfrac{1}{2}$,
   the mechanism is $\sfrac{3}{2}$-consistent and
$\sfrac{3}{2}$-robust.

We return now to the reason that 
we proposed a mechanism that smoothly interpolates
between the \mymaxmin\ mechanism
(\mymaxmintrunc\  with $\gamma=0$)
and
the \mymidornearest\ mechanism
(\mymaxmintrunc\  with $\gamma=\sfrac{1}{2}$).
The reason is that
the \mymaxmin\ mechanism has optimal
consistency,
while the \mymidornearest\ mechanism
has optimal robustness (achieving an optimal
2-approximation of the maximum distance,
and an optimal $\sfrac{3}{2}$-approximation of 
  the minimum utility). 
  Indeed, as we show next, the \mymidornearest\ mechanism
  is {\bf the unique} anonymous, Pareto efficient and
  strategy proof mechanism that
 $\sfrac{3}{2}$-approximates
  the minimum utility.
  
\begin{theorem}
  No anonymous, Pareto efficient
  and strategy-proof mechanism
  besides the \mymidornearest\ mechanism has as good an
  approximation ratio 
  of the minimum utility. \end{theorem}
\begin{proof}
Consider any anonymous, Pareto efficient
  and strategy-proof mechanism. 
This is a median mechanism with $n-1$ phantoms \cite{moulin1980}. 
If this is not the \mymidornearest\ mechanism, one 
of the phantoms will be different to $\sfrac{1}{2}$. 
Consider the smallest such 
  phantom $a$. Suppose $0 \leq a < \sfrac{1}{2}$. A dual argument 
  holds for $\sfrac{1}{2} < a \leq 1$.
  Consider one agent at $1$ and the remaining agents
  at $a$. The facility is located at $a$, giving a 
  minimum utility of
  $a$. The optimal minimum
  utility is $\sfrac{1}{2}+\sfrac{a}{2}$. Therefore the approximation
  ratio is $\frac{(1+a)}{2a}$.
  For $0 \leq a < \sfrac{1}{2}$, this is in $(\sfrac{3}{2},\infty]$.
  Hence 
  the approximation ratio is worse than $\sfrac{3}{2}$.
\end{proof}

\section{Randomized mechanisms}

A randomized mechanism
returns a probability distribution over ex post outcomes. We
compute expectations for approximation ratios, 
robustness and consistency over this distribution.
We can often
achieve better results in expectation with
randomized mechanisms. 
Consider, for example, 
the strategy proof
\mylrm\ mechanism which locates
the facility at $x_1$ with probability $\sfrac{1}{4}$,
at $\frac{(x_1+x_n)}{2}$ with probability $\sfrac{1}{2}$,
and at $x_n$ with probability $\sfrac{1}{4}$.
This achieves an optimal $\sfrac{3}{2}$-approximation of
the maximum distance in expectation
(Procaccia and Tennenholtz \shortcite{ptacmtec2013}
show this for the real line but the result
easily extends to any interval).
The \mylrm\ mechanism does not do quite as well
at approximating the optimal minimum utility, only 
2-approximating it in expectation
\cite{wecai2024}.
\myOmit{
  Indeed, this is worse than a deterministic and
strategy proof mechanism such as 
the   \mymidornearest\ mechanism
which $\sfrac{3}{2}$-approximates
the minimum utility (which is optimal for deterministic and
strategy proof mechanisms). }

The \mylrm\ mechanism can be adapted to take
advantage of predictions. Given a parameter $\delta \in
[0,\sfrac{1}{2}]$, the \mylrmormaxmin\ mechanism proposed
in \cite{abgou22}
uses the \mylrm\ mechanism with probability $2\delta$,
and the \mymaxmin\ mechanism with probability $1-2\delta$.
This achieves an optimal $1+\delta$-consistency and
$2-\delta$-robustness in expectation with respect to the maximum
distance (Proposition 1 and Theorem 1 in \cite{abgou22}).

\begin{theorem}
  For $\delta \in [0,\sfrac{1}{2}]$,
  the \mylrmormaxmin\ mechanism is
  $\frac{1}{(1-\delta)}$-consistent
 and $\frac{1}{\delta}$-robust
  in expectation with respect to the optimal
  minimum utility.
   \end{theorem}
\begin{proof}
  Suppose the prediction is correct and the
  optimal minimum utility $u$.
  The expected minimum utility is $2 \delta \frac{u}{2} + (1-2\delta)
  u  = (1-\delta)u$. Hence it is $\frac{1}{(1-\delta)}$-consistent.
  Suppose the prediction is incorrect. 
  The expected minimum utility is $2 \delta \frac{u}{2}   = \delta u$. Hence it is $\frac{1}{\delta}$-robust.
\end{proof}

For any $\delta > 0$, both the approximation
ratios for consistency and robustness are
worse with respect to 
minimum utility compared to the ratios for maximum
distance. 

As with deterministic mechanisms,
censoring extreme facility locations improves 
performance. 
Let $y=\mymax(\sfrac{1}{3},\mymin(x_1,\sfrac{2}{3}))$
and
$z=\mymax(\sfrac{1}{3},\mymin(\sfrac{2}{3},x_n))$. 
The \mylrmtrunc\  mechanism proposed in
\cite{wecai2024}
locates the facility at $y$ with
probability $\sfrac{1}{4}$,
at $\sfrac{(y+z)}{2}$ with probability $\sfrac{1}{2}$
and $z$ otherwise. 
This mechanism truncates
facility locations to $[\sfrac{1}{3},\sfrac{2}{3}]$. 
It is strategy proof and achieves
in expectation
an optimal $\sfrac{4}{3}$-approximation of the minimum utility,
and a 
$2$-approximation of
the maximum distance 
\cite{wecai2024}. 
The \mylrmtrunc\ mechanism can also be adapted to take
advantage of predictions by combining it with
the \mymaxmin\ mechanism. Given a parameter $\delta \in
[0,\sfrac{1}{2}]$, the \mylrmtruncormaxmin\ 
mechanism 
uses the \mylrmtrunc\ mechanism with probability $2\delta$,
and the \mymaxmin\ mechanism with probability $1-2\delta$.

\begin{theorem}
  For $\delta \in [0,\sfrac{1}{2}]$,
  the \mylrmtruncormaxmin\ mechanism is strategy proof, 
$\frac{2}{(2-\delta)}$-consistent
 and $\frac{2}{3\delta}$-robust
  in expectation with respect to the optimal
  minimum utility.
It is also 
$1+2\delta$-consistent
 and 2-robust
  in expectation with respect to the optimal
  maximum distance. 
   \end{theorem}
   \begin{proof}
     Strategy proofness is immediate from the strategy proofness of
     the constituent mechanisms and the fact that
     the choice of mechanism is independent of
     the agents' reports. 

     Suppose the prediction is correct and the
  optimal minimum utility is $u$.
  The expected minimum utility is $2 \delta \frac{3}{4} u + (1-2\delta)
  u  = (1-\frac{\delta}{2})u$. Hence it is $\frac{2}{(2-\delta)}$-consistent.
  Suppose the prediction is incorrect. 
  The expected minimum utility is $2 \delta \frac{3}{4} u  =
  \frac{3 \delta }{2} u$. Hence it is $\frac{2}{3\delta}$-robust.

  Suppose the prediction is correct and the
  optimal maximum distance is $d$.
  The expected maximum distance is $2 \delta 2 d + (1-2\delta)
  d  = (1+2\delta)d$. Hence it is $1+2\delta$-consistent.
  Suppose the prediction is incorrect. 
  With respect to the maximum distance,
  since it is a probablistic
  mixture of two 2-robust mechanisms,
  it is itself 2-robust in expectation. 
\end{proof}

In approximating the minimum utility,
the \mylrmtruncormaxmin\ mechanism outperforms
the \mylrmormaxmin\ mechanism in four
ways:
\begin{enumerate} \itemsep=0pt
    \item
For any fixed $\delta > 0$,
both the consistency 
and robustness 
of \mylrmtruncormaxmin\
are better than for \mylrmormaxmin. 
\item
For any given consistency in the interval $(1,\sfrac{4}{3}]$,  
\mylrmtruncormaxmin\ achieves
an expected robustness than is three times smaller
than for \mylrmormaxmin. 
%
%
%
%
\item
For any given robustness greater than or equal to $\sfrac{4}{3}$, 
\mylrmtruncormaxmin\ achieves
a smaller expected consistency than \mylrmormaxmin. 
%
%
%
\item
  \mylrmtruncormaxmin\ achieves a consistency in 
  $[1,\sfrac{3}{4}]$, while \mylrmormaxmin\
    achieves a consistency in $[1,2]$. 
\end{enumerate}
On the other hand, in approximating the maximum distance, 
the \mylrmormaxmin\ mechanism outperforms
the \mylrmtruncormaxmin\ mechanism again
in four other ways:
\begin{enumerate} \itemsep=0pt
    \item
For any fixed $\delta > 0$,
both the consistency 
and robustness 
of \mylrmormaxmin\
are better than for \mylrmtruncormaxmin. 
\item
For any given consistency $c \in (1,\sfrac{3}{2}]$, 
\mylrmormaxmin\ achieves
an expected robustness of $3-c$ which is
strictly smaller
than the fixed 2-robustness of \mylrmtruncormaxmin. 
%
%
%
\item
  \mylrmormaxmin\ achieves a robustness in $[\sfrac{3}{2},2]$,
  while
  \mylrmtruncormaxmin\
    is only ever 2-robust.
\item
  \mylrmormaxmin\ achieves a consistency in 
  $[1,\sfrac{3}{2}]$, while \mylrmtruncormaxmin\
    achieves a consistency in $[1,2]$. 
\end{enumerate}
By adjusting
  $\delta$, both mechanisms again {\bf trade 
  consistency for robustness}.
   At $\delta=0$,
   both \mylrmormaxmin\ 
   and
   \mylrmtruncormaxmin\
   are 1-consistent with respect to minimum utility but have unbounded
   robustness. Increasing $\delta$ decreases robustness but
   increases consistency.
   At $\delta=\sfrac{1}{2}$,
   \mylrmormaxmin\ is 2-consistent 
   and 2-robust with respect to minimum utility, while 
   \mylrmtruncormaxmin\ 
   is $\sfrac{4}{3}$-consistent and
$\sfrac{4}{3}$-robust. 
See Figure 1 for a visualization.

  \begin{figure}[htb]
\includegraphics[width=0.5\textwidth,height=190pt,scale=2.4]{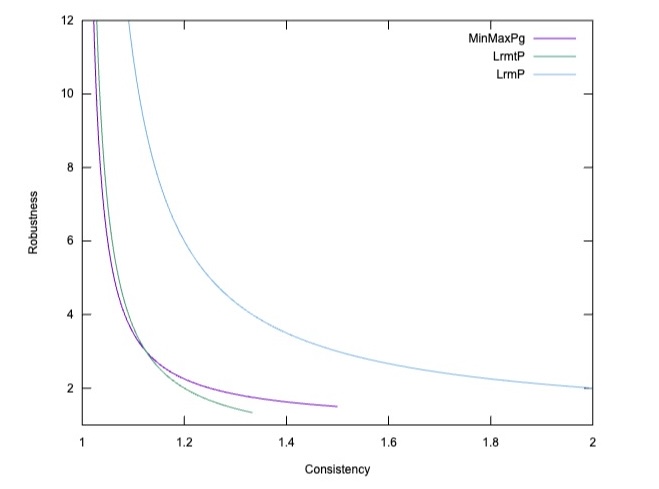}
 \caption{Trade-off between consistency (x-axis) and robustness
   (y-axis) with respect to the minimum utility
   for the \mymaxmintrunc\  mechanism when varying $\gamma \in
   [0,\sfrac{1}{2}]$,
and for the   \mylrmormaxmin\ 
and \mylrmtruncormaxmin\ mechanisms when varying $\delta \in
[0,\sfrac{1}{2}]$. 
}
\end{figure}

\section{Two facilities}

We now design several
new mechanisms for locating two facilities in which
the mechanism
is provided with {\bf two predictions}, one for
the 
location of the leftmost facility, and another for the
location of the rightmost facility.
Xu and Lu \shortcite{ijcai2022p81} propose
a deterministic mechanism with predictions for the two facility
problem that is $(1+\sfrac{n}{2})$-consistent and
$(2n-1)$-robust with respect to the maximum distance.
They observe:
\begin{quote}
{\em ``Whether there is a mechanism with $o(n)$-consistent and
  a bounded robustness is a very interesting open question.''
}
\end{quote}
We answer this open question positively in two ways.
First, we design a novel
deterministic mechanism for approximating the minimum
utility with bounded consistency and robustness.
Second, we design a novel randomized mechanism
for approximating the maximum distance 
with bounded consistency and robustness. 

The \mymaxmintwo\ mechanism 
locates the two facilities by applying the \mymaxmin\
mechanism to each of the predictions in turn. 
If $x_1$ and $x_n$ are the minimum and maximum locations
of the agents, and $\pi_1$ and $\pi_2$ are the two predicted
locations of the facilities,
then $\mymaxmintwo(x_1,x_n,\pi_1,\pi_2)$
locates one facility at $\mymaxmin(x_1,x_n,\pi_1)$
and the other at $\mymaxmin(x_1,x_n,\pi_2)$.
The next theorem demonstrates that 
the \mymaxmintwo\ mechanism is 1-consistent with respect
to the maximum distance or minimum utility,
$\sfrac{3}{2}$-robust with respect to the minimum utility,
but has unbounded robustness with respect to the
maximum distance.

As with one facility, we also adapt the mechanism
to censor extreme predictions. This again lets us
{\bf trade consistency for robustness}.
Given $\lambda \in [0,\sfrac{1}{4}]$,
\mymaxmintwotrunc$(x_1,x_n,\pi_1,\pi_2)$
maps the leftmost prediction
$\pi_1$ onto $\pi_1' = \mymax(\lambda,\mymin(\pi,1-3\lambda))$, the
rightmost prediction
$\pi_2$ onto $\pi_2' = \mymax(3\lambda,\mymin(\pi,1-\lambda))$, and then
applies 
\mymaxmintwo$(x_1,x_n,\pi_1',\pi_2')$.
For $\lambda = 0$, predictions are not censored. 
For $\lambda = \sfrac{1}{4}$,
the leftmost prediction is mapped onto $\sfrac{1}{4}$
while the rightmost prediction is mapped onto
$\sfrac{3}{4}$.
More generally, the leftmost prediction is mapped
into $[\lambda,1-3\lambda]$,
and the rightmost prediction into $[3\lambda,1-\lambda]$.

\begin{theorem}
  For $\lambda \in [0,\sfrac{1}{4}]$,
  the \mymaxmintwotrunc\ mechanism is strategy proof,
$\frac{(2-\lambda)}{(2-2\lambda)}$-consistent
  and $\frac{(3+2\lambda)}{2(1+2\lambda)}$-robust 
  with respect to the optimal
  minimum utility.
  At $\lambda=0$, it is $1$-consistent and $\sfrac{3}{2}$-robust while
  at $\lambda=\sfrac{1}{4}$, it is $\sfrac{7}{6}$-consistent and
$\sfrac{7}{6}$-robust.

With respect to the optimal maximum
  distance, 
\mymaxmintwotrunc\ is $1$-consistent and has
  unbounded robustness 
  at $\lambda=0$, 
  and has unbounded consistency and robustness
for $\lambda>0$. 
   \end{theorem}
\begin{proof}
Strategy proofness is immediate from that of the untruncated
mechanism.
With respect to the optimal minimum utility,
suppose the two predictions are correct.
Using a similar case analysis to Theorem 1,
the worst case is when the optimal facility location
is halfway between 0 and $\lambda$, and agents served
by this facility are in $[0,\lambda]$ including at the endpoints
of the interval. Suppose in this case that 
the optimal minimum utility is $1-b$. Then 
$\pi_1 = b = \sfrac{\lambda}{2}$. 
The minimum utility of the solution returned by
the \mymaxmintwotrunc\ mechanism is $1-2b$, 
giving an approximation
ratio of $\frac{(1-b)}{(1-2b)} = \frac{(2-\lambda)}{(2-2\lambda)}$. 
The 
mechanism is therefore
$\frac{(2-\lambda)}{(2-2\lambda)}$-consistent. 

Now we consider the possibility that the predictions are incorrect.
The leftmost truncated prediction is in
$[\lambda,1-3\lambda]$, and the rightmost
prediction in $[3\lambda,1-\lambda]$. 
Using a similar case analysis, 
the worst case is when the facilities are located
by the mechanism at their most extreme point (i.e.
$x_1=\lambda$, $x_n=1-\lambda$). 
The optimal minimum utility in this
setting is $\frac{(3+2\lambda)}{4}$,
while the minimum utility of the solution returned by
the \mymaxmintwotrunc\ mechanism is 
$\frac{(1+2\lambda)}{2}$. 
This gives an approximation
ratio of 
$\frac{(3+2\lambda)}{2(1+2\lambda)}$.
The \mymaxmintrunc\ mechanism is therefore
$\frac{(3+2\lambda)}{2(1+2\lambda)}$-robust with respect to the
minimum utility.

With respect to the optimal maximum distance,
suppose the predictions are correct. For $\lambda=0$,
the \mymaxmintwotrunc\ mechanism is equivalent to
\mymaxmintwo\ mechanism.
It is easy to see that this is $1$-consistent.
For $\lambda>0$,
consider agents at 0 and 1.
The optimal maximum distance is zero,
but 
the \mymaxmintwotrunc\ mechanism locates
facilities at $\lambda$ and $1-\lambda$, giving
a maximum distance of $\lambda$.
The consistency is therefore unbounded. 
Similarly the robustness is unbounded with respect
to the maximum distance irrespective of $\lambda$. 
\end{proof}

As with locating a single facility,
we can achieve better consistency and robustness in expectation
with randomized mechanisms. 
The
\myem\ mechanism (called Mechanism 2  by
Procaccia and Tennenholtz \shortcite{ptacmtec2013})
  locates two facilities in three ways:
  (1) at $x_1$ and $x_n$ with probability
  $\sfrac{1}{2}$; 
  (2) at $x_1+2d$ and $x_n-2d$ 
with probability
$\sfrac{1}{6}$ where
$d$ is the optimal minimum distance;
  (3) and
at $x_1 + d$ and $x_n-d$ with the
remaining probability $\sfrac{1}{3}$. 
The 
mechanism is strategy proof and $\sfrac{5}{3}$-approximates
the maximum distance in expectation
\cite{ptacmtec2013}. It achieves an even better approximation
ratio of the minimum utility.

\begin{theorem}
  The \myem\ mechanism 
$\sfrac{9}{7}$-approximates the optimal
  minimum utility in expectation.
   \end{theorem}
   \begin{proof}
     With probability $\sfrac{1}{3}$, the mechanism has
     a minimum utility $1-d$, and with probability $\sfrac{2}{3}$, it
     has
     a minimum utility of $1-2d$. The expected minimum utility is
     thus $\frac{(1-d+2-4d)}{3} = \frac{(3-5d)}{3}$.
     This compares to an optimal minimum utility of $1-d$.
     The approximation ratio is thus $\frac{3(1-d)}{(3-5d)}$.
     This is maximized for $d=\sfrac{1}{4}$ when it is
     $\sfrac{9}{7}$.  Hence, the mechanism
     $\sfrac{9}{7}$-approximates the minimum utility. 
     \end{proof}

     The randomized \myem\ mechanism can be augmented to take
     advantage of two predictions for the  optimal
     locations of the two facilities.
     Given a parameter $\theta \in
[0,\sfrac{1}{2}]$, the \myemormaxmintwo\ 
uses 
\myem\ 
with probability $2\theta$,
and 
\mymaxmintwo\ 
with probability $1-2\theta$.

\begin{theorem}
The \myemormaxmintwo\ mechanism is strategy proof.
  With respect to the optimal
  minimum utility, 
it is $\frac{9}{(9-4\theta)}$-consistent
and $\frac{9}{2(3+\theta)}$-robust
for $\theta \in
[0,\sfrac{1}{2}]$. 
  With respect to the optimal maximum
  distance, it 
  is $\frac{(3+4\theta)}{3}$-consistent
  for $\theta \in
[0,\sfrac{1}{2}]$, 
  $\sfrac{5}{3}$-robust 
  for $\theta=\sfrac{1}{2}$. 
  and has unbounded robustness for $\theta < \sfrac{1}{2}$.
   \end{theorem}
\begin{proof}
     Strategy proofness is immediate from the strategy proofness of
     the constituent mechanisms and the fact that
     the choice of mechanism is independent of
     the agents' reports. 

    Suppose the prediction is correct and the
  optimal minimum utility is $u$.
  The expected minimum utility is $2 \theta \frac{7}{9} u + (1-2\theta)
  u  = (1-\frac{4\theta}{9})u$. Hence it is $\frac{9}{(9-4\theta)}$-consistent.
  Suppose the prediction is incorrect. 
  The expected minimum utility is $2 \theta \frac{7}{9} u +
  (1-2\theta) \frac{2}{3} u  =
  \frac{2(3+\theta)}{9} u$. Hence it is $\frac{9}{2(3+\theta)}$-robust.

  Suppose the prediction is correct and the
  optimal maximum distance is $d$.
  The expected maximum distance is $2 \theta \frac{5}{3} d +
  (1-2\theta) d  =
  \frac{(3+4\theta)}{3} d$. Hence it is $\frac{(3+4\theta)}{3}$-consistent.
     \end{proof}

\myOmit{
\begin{figure}[htb]
\includegraphics[scale=1.2,width=0.6\textwidth,height=190pt]{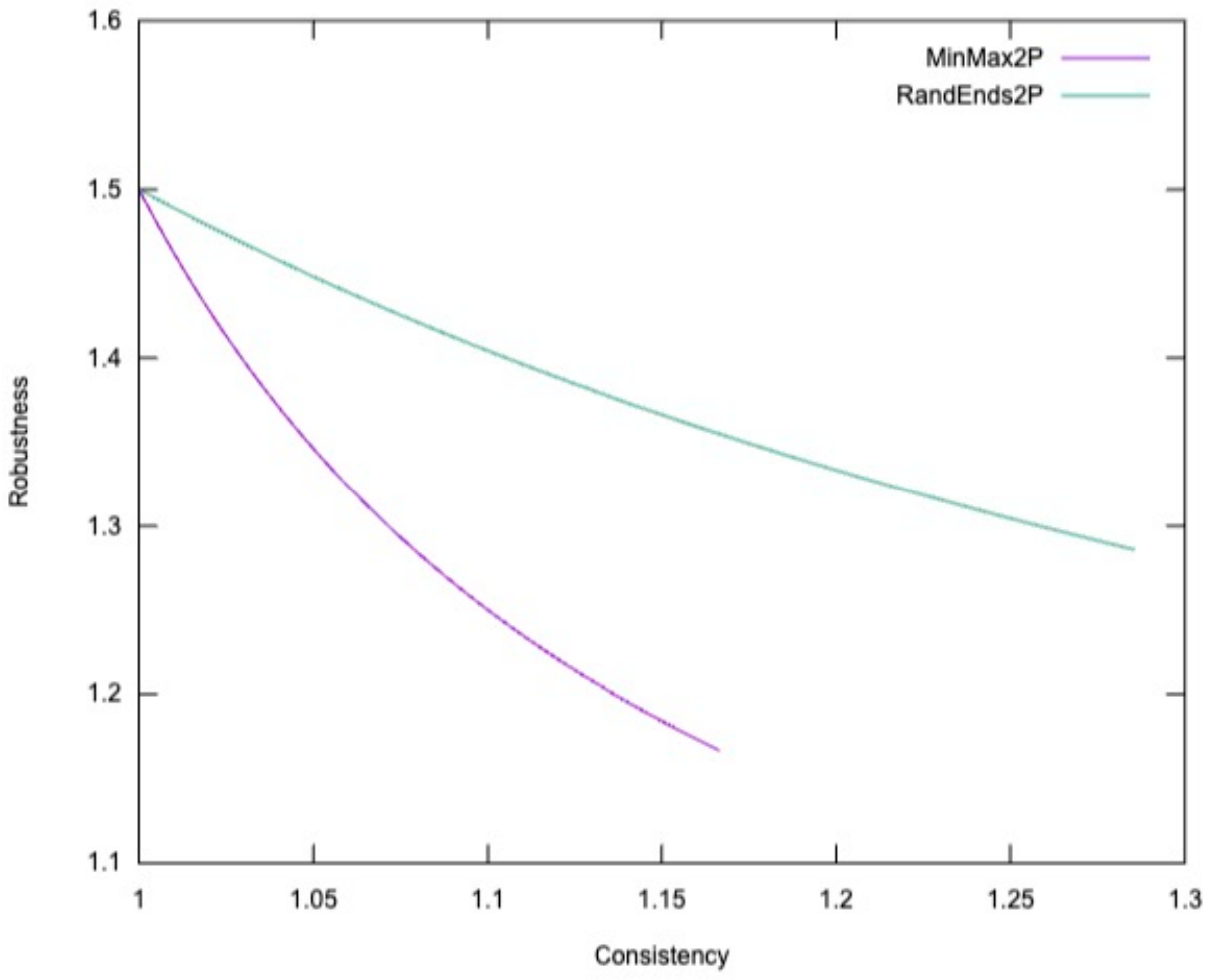}
\caption{Trade-off between consistency (x-axis) and robustness
   (y-axis) for the \mymaxmintwotrunc\ mechanism for $\lambda \in [0,
   \sfrac{1}{4}]$, and for the \myemormaxmintwo\ mechanism 
  for $\theta \in [0,\sfrac{1}{2}]$. 
   At $\lambda=0$, the \mymaxmintwotrunc\
   mechanism
   is 1-consistent and $\sfrac{3}{2}$-robust
   with respect to minimum utility. 
At $\lambda=\sfrac{1}{4}$, it is $\sfrac{7}{6}$-consistent and
$\sfrac{7}{6}$-robust.
   At $\theta=0$, the \myemormaxmintwo\
   mechanism
   is 1-consistent and $\sfrac{3}{2}$-robust
   with respect to minimum utility. 
At $\theta=\sfrac{1}{2}$, it is $\sfrac{9}{7}$-consistent and
$\sfrac{9}{7}$-robust.
}
\end{figure}
}

     This mechanism again lets us {\bf trade
     consistency for robustness}.
At $\theta=0$, 
the \myemormaxmintwo\ mechanism 
   is 1-consistent and $\sfrac{3}{2}$-robust
   with respect to minimum utility.
Increasing $\theta$ decreases robusness but
increases consistency.
At $\theta=\sfrac{1}{2}$, it is $\sfrac{9}{7}$-consistent and
$\sfrac{9}{7}$-robust.

\section{Characterization of consistent mechanisms}

You might wonder why we have mostly considered 
mechanisms that are based on (truncated versions of)
\mymaxmin. 
We now give a result that characeterizes strategy proof mechanisms
with predictions which achieve good levels of consistency.
This characterization result demonstrates the central
role played by
the \mymaxmin\ mechanism. We say that a prediction is {\em extreme} iff it is 0 or 1,
and {\em non-extreme} otherwise.
With extreme predictions, there are multiple
mechanisms that are 1-consistent
(e.g. \mymidornearest, \mymedian\ and \mymaxmin).
With non-extreme predictions,
the only mechanism better than 2-consistent
with respect to the maximum distance
is 
\mymaxmin. 

\begin{theorem}
  For non-extreme predictions,
  the only deterministic, strategy proof, anonymous and Pareto
  efficient mechanism using predictions
  that is better than $2$-consistent with respect to the
maximum distance is  \mymaxmin\ which is
$1$-consistent. It is also the only such mechanism
that is $1$-consistent with respect to the minimum utility. 
\end{theorem}
\begin{proof}
Any such mechanism is a generalized
median mechanism with $n-1$ phantoms \cite{moulin1980}. 
If the phantoms are all at the predicted facility location, then we have
the \mymaxmin\ mechanism.
Suppose instead that
one or more of those phantoms is not at the predicted and correct facility
location $\pi$.  
And suppose $\pi \leq \sfrac{1}{2}$. There is
a dual argument for $\pi \geq \sfrac{1}{2}$.
By the non-extreme assumption $\pi > 0$. 
%
%
%
%
Let $\rho$ be the largest such phantom different to $\pi$.
There are three cases.
In the first case $\rho < \pi$.
Consider $n-1$ agents at $\rho$ and one at $2\pi -
\rho$. As $0 < \pi \leq \sfrac{1}{2}$, $2\pi -\rho$ is
in $(\pi,1]$.
The optimal facility location is
$\pi$ as required.
However, the mechanism locates the
facility at $\rho$ which is twice 
the optimal maximum
distance from the agent at $2\pi - \rho$.
In the second case $\pi < \rho \leq 2\pi$.
Consider one agent at $2\pi - \rho$, and the other $n-1$ agents
at $\rho$. Note that $2\pi -\rho$ is in $[0,\pi)$. 
The optimal facility location is
$\pi$ as required.
However, the mechanism locates the
facility at $\rho$ which is twice 
the optimal maximum
distance from the agent at $2\pi - \rho$.
In the third case $2\pi < \rho$.
Consider one agent at 0, and the other $n-1$ agents at $2\pi$.
Note that $2\pi$ is
in $(0,1]$.
The facility is located at $2\pi$
at twice the optimal maximum distance
from the agent at 0.
In each case, the mechanism is
2-consistent with respect to the maximum distance. 
By a similar argument, 
\mymaxmin\
is the only such mechanism
that is 1-consistent with respect to the minimum utility.
\end{proof}

  On the (unbounded) real line, predictions
  are never extreme and \mymaxmin\ is the unique
  deterministic mechanism
  achieving better than 2-consistency.

\section{Conclusions}

\begin{table}[htb]
\hspace{-0.3cm}  \begin{tabular}{|c|c|c|c|c|} \hline
& \multicolumn{2}{|c|} {\bf max distance} & \multicolumn{2}{|c|} {\bf min
                                        utility} \\ \hline
1 facility, determin &  {\bf consis} & {\bf robust} & {\bf
                                                            consis} &
                                                                       {\bf robust} \\ \hline
lower bound & 1  & 2 & 1 & $\sfrac{3}{2}$ \\
    \mymaxmin   & 1 & 2 & {\bf 1}  & $\pmb{\infty}$ \\
    \mymidornearest & {\bf 2} & {\bf 2} & $\pmb{\sfrac{3}{2}}$ & $\pmb{\sfrac{3}{2}}$ \\
    \mymaxmintrunc, $\gamma>0$ & {\bf 2} & {\bf 2} & $\pmb{\frac{(2-\gamma)}{(2-2\gamma)}}$ &
                                                               $\pmb{\frac{(1+\gamma)}{2\gamma}}$ \\ 
    \mymaxmintrunc, $\gamma=\sfrac{1}{2}$ & {\bf 2} & {\bf 2} & $\pmb{\sfrac{3}{2}}$ &
                                                               $\pmb{\sfrac{3}{2}}$ \\ \hline
1 facility, random &  &  &  & \\ \hline
lower bound &  1 & $\sfrac{3}{2}$ & 1 & $\sfrac{4}{3}$ \\
\mylrmormaxmin         &  $1+\delta$ & $2-\delta$ & $\pmb{\frac{1}{(1-\delta)}}$ & $\pmb{\frac{1}{\delta}}$ \\
\mylrmormaxmin, $\delta=\sfrac{1}{2}$         &  $\sfrac{3}{2}$ & $\sfrac{3}{2}$ & $\pmb{2}$ & $\pmb{2}$ \\
    \mylrmtruncormaxmin\             & $\pmb{1+2\delta}$ & {\bf 2} & $\pmb{\frac{2}{(2-\delta)}}$ &
                                                                      $\pmb{\frac{2}{3\delta}}$
                   \\
    \mylrmtruncormaxmin, $\delta=\sfrac{1}{2}$             & $\pmb{2}$ & {\bf 2} & $\pmb{\sfrac{4}{3}}$ &
                                                                      $\pmb{\sfrac{4}{3}}$
                   \\
                   \hline
                   2 facilities, determin &  &  &  & \\ \hline
                   lower bound & 1 & $n-2$ & 1 & $\sfrac{10}{9}$ \\ 
\mymaxmintwo            & {\bf 1} & $\pmb{\infty}$ &
                                        {\bf 1}
                         & $\pmb{\sfrac{3}{2}}$ \\
\mymaxmintwotrunc, $\lambda > 0$              & $\pmb{\infty}$ & $\pmb{\infty}$ &
                                        $\pmb{\frac{(2-\lambda)}{(2-2\lambda)}}$
                         & $\pmb{\frac{(3+2\lambda)}{2(1+2\lambda)}}$ \\
\mymaxmintwotrunc, $\lambda = \sfrac{1}{4}$              & $\pmb{\infty}$ & $\pmb{\infty}$ &
                                        $\pmb{\sfrac{7}{6}}$
                         & $\pmb{\sfrac{7}{6}}$ \\
   \hline 
2 facilities, random &  &  &  & \\ \hline
                   lower bound & 1 & $\sfrac{3}{2}$ & 1 & $\sfrac{10}{9}$ \\ 
    \myemormaxmintwo, $\theta = 0$    & $\pmb{1}$ &  $\pmb{\infty}$ &
                                                   $\pmb{1}$
                         & $\pmb{\sfrac{3}{2}}$ \\ 
    \myemormaxmintwo, $\theta < \sfrac{1}{2}$    & $\pmb{\frac{(3+4\theta)}{3}}$ &  $\pmb{\infty}$ &
                                                   $\pmb{\frac{9}{(9-4\theta)}}$
                         & $\pmb{\frac{9}{2(3+\theta)}}$ \\ 
    \myemormaxmintwo, $\theta = \sfrac{1}{2}$    & $\pmb{\sfrac{5}{3}}$ & $\pmb{\sfrac{5}{3}}$ &
                                                   $\pmb{\sfrac{9}{7}}$
                        & $\pmb{\sfrac{9}{7}}$ \\ \hline
\end{tabular}
\caption{Summary of {\bf consis}tency and
  {\bf robust}ness results with respect to the
optimal  {\bf max}imum {\bf distance} or {\bf min}imum {\bf utility}
for {\bf determin}istic or {\bf random}ized
strategy proof mechanisms. {\bf Bold} font for
results proved here.
}
\end{table}

Our examination of mechanisms for facility location augmented
with predictions of the optimal location demonstrates
that an egalitarian
viewpoint
considering {\em both} the maximum distance any agent travels {\em
  and} the least
utility of any agent provides a more complex picture of 
performance than one considering just maximum distance alone.
Our results are summarized in Table 1.
By considering how mechanisms can perform poorly, we proposed
new deterministic and randomized mechanisms
for locating a single
facility that achieve bounded robustness with respect to both
  maximum distance and minimum utility.
For locating two facilities, 
  we also designed novel mechanisms with predictions
  with bounded robustness and
  consistency.
  These new mechanisms let us smoothly trade consistency
  for robustness.
    A repeated idea to obtain
  good performance was to censor extreme predictions. 
  \myOmit{
  For the 2-d facility location problem,
  we obtained similar results.
  For example, the \myboundbox\ mechanism which applies the \mymaxmin\
  mechanism
  along each dimension 
  is 1-consistent and $1+\sqrt{2}$-robust with respect to the maximum
  distance, and while it is
  1-consistent with respect to the minimum utility, there is no bound
  on its robustness. 
  We proposed instead the \myboundboxtrunc\ mechanism
  which applies the \mymaxmintrunc\
  mechanism
  along each dimension. 
  This again achieves bounded robustness with respect to both
  maximum distance and minimum utility.

  We also provided some characterization results on consistency. For instance,
  on the 1-d problem, the only deterministic, strategy proof
  and anonymous mechanism that is better than 2-consistent with
  respect to the
  maximum distance is the \mymaxmin\ mechanism which is 1-consistent.
  This characterization result lifts in a straight forward way
  to the 2-d problem. 

}

\bibliographystyle{named}
\bibliography{/Users/z3193295/Documents/biblio/a-z2,/Users/z3193295/Documents/biblio/pub2}

\begin{thebibliography}{}

\bibitem[\protect\citeauthoryear{Agrawal \bgroup \em et al.\egroup
  }{2022}]{abgou22}
Priyank Agrawal, Eric Balkanski, Vasilis Gkatzelis, Tingting Ou, and Xizhi Tan.
\newblock Learning-augmented mechanism design: Leveraging predictions for
  facility location.
\newblock In David~M. Pennock, Ilya Segal, and Sven Seuken, editors, {\em {EC}
  '22: The 23rd {ACM} Conference on Economics and Computation, Boulder, CO,
  USA, July 11 - 15, 2022}, pages 497--528. {ACM}, 2022.

\bibitem[\protect\citeauthoryear{Almanza \bgroup \em et al.\egroup
  }{2021}]{onlineflpm}
Matteo Almanza, Flavio Chierichetti, Silvio Lattanzi, Alessandro Panconesi, and
  Giuseppe Re.
\newblock Online facility location with multiple advice.
\newblock In Marc'Aurelio Ranzato, Alina Beygelzimer, Yann~N. Dauphin, Percy
  Liang, and Jennifer~Wortman Vaughan, editors, {\em Advances in Neural
  Information Processing Systems 34: Annual Conference on Neural Information
  Processing Systems 2021, NeurIPS 2021, December 6-14, 2021, virtual}, pages
  4661--4673, 2021.

\bibitem[\protect\citeauthoryear{Aziz \bgroup \em et al.\egroup
  }{2020}]{acllwaaai2020}
H.~Aziz, H.~Chan, B.E. Lee, B.~Li, and T.~Walsh.
\newblock Facility location problem with capacity constraints: Algorithmic and
  mechanism design perspectives.
\newblock In Vincent Conitzer and Fei Sha, editors, {\em Proceedings of the
  Thirty-Fourth {AAAI} Conference on Artificial Intelligence}. {AAAI} Press,
  2020.

\bibitem[\protect\citeauthoryear{Aziz \bgroup \em et al.\egroup
  }{2021}]{abs-2111-01566}
Haris Aziz, Alexander Lam, Barton~E. Lee, and Toby Walsh.
\newblock Strategyproof and proportionally fair facility location.
\newblock {\em CoRR}, abs/2111.01566, 2021.

\bibitem[\protect\citeauthoryear{Aziz \bgroup \em et al.\egroup
  }{2022}]{propwine2022}
Haris Aziz, Alexander Lam, Barton~E. Lee, and Toby Walsh.
\newblock Strategyproof and proportionally fair facility location.
\newblock In Kristoffer~Arnsfelt Hansen, Tracy~Xiao Liu, and Azarakhsh
  Malekian, editors, {\em Web and Internet Economics - 18th International
  Conference, {WINE} 2022, Troy, NY, USA, December 12-15, 2022, Proceedings},
  volume 13778 of {\em Lecture Notes in Computer Science}, page 351. Springer,
  2022.

\bibitem[\protect\citeauthoryear{Chan \bgroup \em et al.\egroup
  }{2021}]{ijcai2021-596}
Hau Chan, Aris Filos-Ratsikas, Bo~Li, Minming Li, and Chenhao Wang.
\newblock Mechanism design for facility location problems: A survey.
\newblock In Zhi-Hua Zhou, editor, {\em Proceedings of the Thirtieth
  International Joint Conference on Artificial Intelligence, {IJCAI-21}}, pages
  4356--4365. International Joint Conferences on Artificial Intelligence
  Organization, 8 2021.
\newblock Survey Track.

\bibitem[\protect\citeauthoryear{Cheng and Zhou}{2015}]{faclocsurvey}
Yukun Cheng and Sanming Zhou.
\newblock A survey on approximation mechanism design without money for facility
  games.
\newblock In David Gao, Ning Ruan, and Wenxun Xing, editors, {\em Advances in
  Global Optimization}, pages 117--128, Cham, 2015. Springer International
  Publishing.

\bibitem[\protect\citeauthoryear{Escoffier \bgroup \em et al.\egroup
  }{2011}]{egktps2011}
B.~Escoffier, L.~Gourv{\`e}s, N.~Kim Thang, F.~Pascual, and O.~Spanjaard.
\newblock Strategy-proof mechanisms for facility location games with many
  facilities.
\newblock In R.I. Brafman, F.S Roberts, and A.~Tsouki{\`a}s, editors, {\em
  Algorithmic Decision Theory}, pages 67--81, Berlin, Heidelberg, 2011.
  Springer Berlin Heidelberg.

\bibitem[\protect\citeauthoryear{Fotakis and Tzamos}{2010}]{ft2010}
D.~Fotakis and C.~Tzamos.
\newblock Winner-imposing strategyproof mechanisms for multiple facility
  location games.
\newblock In A.~Saberi, editor, {\em Internet and Network Economics}, pages
  234--245, Berlin, Heidelberg, 2010. Springer Berlin Heidelberg.

\bibitem[\protect\citeauthoryear{Fotakis and Tzamos}{2013}]{ft2013}
D.~Fotakis and C.~Tzamos.
\newblock On the power of deterministic mechanisms for facility location games.
\newblock In F.V. Fomin, R.~Freivalds, M.~Kwiatkowska, and D.~Peleg, editors,
  {\em Automata, Languages, and Programming}, pages 449--460, Berlin,
  Heidelberg, 2013. Springer Berlin Heidelberg.

\bibitem[\protect\citeauthoryear{Fotakis \bgroup \em et al.\egroup
  }{2021}]{onlineflpf}
Dimitris Fotakis, Evangelia Gergatsouli, Themis Gouleakis, and Nikolas Patris.
\newblock Learning augmented online facility location.
\newblock {\em CoRR}, abs/2107.08277, 2021.

\bibitem[\protect\citeauthoryear{Goel and Hann{-}Caruthers}{2023}]{coordmedian}
Sumit Goel and Wade Hann{-}Caruthers.
\newblock Optimality of the coordinate-wise median mechanism for strategyproof
  facility location in two dimensions.
\newblock {\em Soc. Choice Welf.}, 61(1):11--34, 2023.

\bibitem[\protect\citeauthoryear{Han \bgroup \em et al.\egroup
  }{2023}]{DBLP:conf/aaai/HanJA23}
Yue Han, Christopher Jerrett, and Elliot Anshelevich.
\newblock Optimizing multiple simultaneous objectives for voting and facility
  location.
\newblock In Brian Williams, Yiling Chen, and Jennifer Neville, editors, {\em
  Thirty-Seventh {AAAI} Conference on Artificial Intelligence, {AAAI} 2023,
  Thirty-Fifth Conference on Innovative Applications of Artificial
  Intelligence, {IAAI} 2023, Thirteenth Symposium on Educational Advances in
  Artificial Intelligence, {EAAI} 2023, Washington, DC, USA, February 7-14,
  2023}, pages 5665--5672. {AAAI} Press, 2023.

\bibitem[\protect\citeauthoryear{Istrate and Bonchis}{2022}]{oflp}
Gabriel Istrate and Cosmin Bonchis.
\newblock Mechanism design with predictions for obnoxious facility location.
\newblock {\em CoRR}, abs/2212.09521, 2022.

\bibitem[\protect\citeauthoryear{Jiang \bgroup \em et al.\egroup
  }{2022}]{onlineflp}
Shaofeng~H.{-}C. Jiang, Erzhi Liu, You Lyu, Zhihao~Gavin Tang, and Yubo Zhang.
\newblock Online facility location with predictions.
\newblock In {\em The Tenth International Conference on Learning
  Representations, {ICLR} 2022, Virtual Event, April 25-29, 2022}.
  OpenReview.net, 2022.

\bibitem[\protect\citeauthoryear{Lu \bgroup \em et al.\egroup
  }{2010}]{proportional}
P.~Lu, X.~Sun, Y.~Wang, and Z.A. Zhu.
\newblock Asymptotically optimal strategy-proof mechanisms for two-facility
  games.
\newblock In {\em Proceedings of the 11th ACM Conference on Electronic
  Commerce}, EC '10, pages 315--324, New York, NY, USA, 2010. ACM.

\bibitem[\protect\citeauthoryear{Mei \bgroup \em et al.\egroup
  }{2016}]{flprevisit}
Lili Mei, Minming Li, Deshi Ye, and Guochuan Zhang.
\newblock Strategy-proof mechanism design for facility location games:
  Revisited (extended abstract).
\newblock In Catholijn~M. Jonker, Stacy Marsella, John Thangarajah, and Karl
  Tuyls, editors, {\em Proceedings of the 2016 International Conference on
  Autonomous Agents {\&} Multiagent Systems, Singapore, May 9-13, 2016}, pages
  1463--1464. {ACM}, 2016.

\bibitem[\protect\citeauthoryear{Mei \bgroup \em et al.\egroup
  }{2019}]{flpdesire}
Lili Mei, Minming Li, Deshi Ye, and Guochuan Zhang.
\newblock Facility location games with distinct desires.
\newblock {\em Discrete Applied Mathematics}, 264:148--160, 2019.

\bibitem[\protect\citeauthoryear{Mitzenmacher and
  Vassilvitskii}{2020}]{predictionsurvey}
Michael Mitzenmacher and Sergei Vassilvitskii.
\newblock Algorithms with predictions.
\newblock In Tim Roughgarden, editor, {\em Beyond the Worst-Case Analysis of
  Algorithms}, pages 646--662. Cambridge University Press, 2020.

\bibitem[\protect\citeauthoryear{Moulin}{1980}]{moulin1980}
H.~Moulin.
\newblock On strategy-proofness and single peakedness.
\newblock {\em Public Choice}, 35(4):437--455, 1980.

\bibitem[\protect\citeauthoryear{Procaccia and
  Tennenholtz}{2009}]{approxmechdesign2}
Ariel~D. Procaccia and Moshe Tennenholtz.
\newblock Approximate mechanism design without money.
\newblock In John Chuang, Lance Fortnow, and Pearl Pu, editors, {\em
  Proceedings 10th {ACM} Conference on Electronic Commerce (EC-2009), Stanford,
  California, USA, July 6--10, 2009}, pages 177--186. {ACM}, 2009.

\bibitem[\protect\citeauthoryear{Procaccia and
  Tennenholtz}{2013}]{ptacmtec2013}
A.D. Procaccia and M.~Tennenholtz.
\newblock Approximate mechanism design without money.
\newblock {\em ACM Trans. Econ. Comput.}, 1(4):18:1--18:26, December 2013.

\bibitem[\protect\citeauthoryear{Tang \bgroup \em et al.\egroup
  }{2020}]{flplimit2}
Zhongzheng Tang, Chenhao Wang, Mengqi Zhang, and Yingchao Zhao.
\newblock Mechanism design for facility location games with candidate
  locations.
\newblock Technical report, CoRR archive within arXiv.org, Cornell University
  Library, 2020.

\bibitem[\protect\citeauthoryear{Walsh}{2021}]{wpricai21}
Toby Walsh.
\newblock Strategy proof mechanisms for facility location at limited locations.
\newblock In Duc~Nghia Pham, Thanaruk Theeramunkong, Guido Governatori, and
  Fenrong Liu, editors, {\em {PRICAI} 2021: Trends in Artificial Intelligence -
  18th Pacific Rim International Conference on Artificial Intelligence}, volume
  13031 of {\em Lecture Notes in Computer Science}, pages 113--124. Springer,
  2021.

\bibitem[\protect\citeauthoryear{Walsh}{2024}]{wecai2024}
Toby Walsh.
\newblock Approximate mechanism design for facility location with multiple
  objectives.
\newblock In {\em Proceedings of ECAI 2024}, Frontiers in Artificial
  Intelligence and Applications. {IOS} Press, 2024.

\bibitem[\protect\citeauthoryear{Xu and Lu}{2022}]{ijcai2022p81}
Chenyang Xu and Pinyan Lu.
\newblock Mechanism design with predictions.
\newblock In Lud~De Raedt, editor, {\em Proceedings of the Thirty-First
  International Joint Conference on Artificial Intelligence, {IJCAI-22}}, pages
  571--577. International Joint Conferences on Artificial Intelligence
  Organization, 7 2022.

\bibitem[\protect\citeauthoryear{Zhang and Li}{2014}]{zhang2014}
Q.~Zhang and M.~Li.
\newblock Strategyproof mechanism design for facility location games with
  weighted agents on a line.
\newblock {\em Journal of Combinatorial Optimization}, 28(4):756--773, Nov
  2014.

\end{thebibliography}

\end{document}